\documentclass[10pt,twocolumn,a4paper,conference]{IEEEtran}

\usepackage{cite}	
\usepackage{graphicx} 
\usepackage[latin1]{inputenc} 
\usepackage[T1]{fontenc} 
\usepackage{amsmath,amsfonts,amsbsy,amssymb} 
\usepackage{mathabx} 
\usepackage{amssymb} 
\usepackage{amsmath}
\usepackage{amsthm}	
\usepackage{mathrsfs}
\usepackage[nolist]{acronym} 
\usepackage{tabularx} 
\usepackage{multirow}
\usepackage{wasysym}
\usepackage{float}
\usepackage{color} 

\usepackage{enumitem} 

\newtheorem{lemma}{Lemma}

\begin{document}

\title{On the Effective Energy Efficiency of Ultra-reliable Networks in the Finite Blocklength Regime} 
\author{Mohammad Shehab, Endrit Dosti, Hirley Alves, and Matti Latva-aho\\
	
	\IEEEauthorblockA{
		Centre for Wireless Communications (CWC), University of Oulu, Finland\\
	}
	Email: firstname.lastname@oulu.fi
}
%
\maketitle


\begin{abstract}
Effective Capacity (EC) indicates the maximum communication rate subject to a certain delay constraint while effective energy efficiency (EEE) denotes the ratio between EC and power consumption. In this paper, we analyze the EEE of ultra-reliable networks operating in the finite blocklength regime. We obtain a closed form approximation for the EEE in Rayleigh block fading channels as a function of power, error probability, and delay. We show the optimum power allocation strategy for maximizing the EEE in finite blocklength transmission which reveals that Shannon's model underestimates the optimum power when compared to the exact finite blocklength model. Furthermore, we characterize the buffer constrained EEE maximization problem for different power consumption models. The results show that accounting for empty buffer probability (EBP) and extending the maximum delay tolerance jointly enhance the EC and EEE. 
\end{abstract}


\section{Introduction}\label{introduction}
The new generation of mobile communication is expected to support a multitude of smart devices interconnected via machine-to-machine (M2M) type networks, enabling the Internet of Things (IoT). Energy efficient transmission while guaranteeing quality-of-service (QoS) is an ultimate goal in the design of 5G. QoS constraints ranging from low latency in the order of few milliseconds and packet loss rate (PLR) ($< 10^{-3}$) are key requirements for \textit{Ultra-Reliable Low Latency Communication} (URLLC) \cite{paper1,Dosti,NokiacMTC2016}. In order to boost throughput and reliability while guaranteeing low latency, it becomes crucial to investigate and optimize the resources that are allocated for transmission. In most cases, URLL application devices have limited power resources which dictates careful planning of throughput maximization with wise energy consumption models. Furthermore, the information and communication technology industry is estimated to contribute to 6$\%$ of global CO$_2$ emission by 2020 \cite{green}. This urges the invention of low power consumption, green communication schemes.

To satisfy extremely low latency in real time applications and emerging technologies such as e-health, wireless sensor networks, and smart grids, an attractive solution is communication with short blocklength messages \cite{paper1}. When the packets are short and delay requirements are stringent,  performance metrics, such as Shannon capacity or outage capacity, provide a poor benchmark \cite{paper2, paper14}. Therefore, fundamentally new approaches are needed \cite {paper5,paper1}. The maximum achievable rate of finite blocklength packets is defined in \cite{paper14} in terms of blocklength and error probability.

The effective capacity (EC) metric was first introduced in \cite{paper6} to guarantee statistical QoS requirements by capturing the physical and link layers aspects. EC maps the maximum arrival rate that can be supported by a network with a maximum delay bound of $\delta$ and a delay outage probability of $\Lambda$. In \cite{paper5}, Gursoy characterized the EC in bits per channel use (bpcu) for short packets in quasi-static fading channels where the channel coefficients remain constant for the whole time spanning one packet transmission. A cognitive transmission scheme was analyzed in \cite{paper9}, where the target was to maximize the secondary user EC with constraints on the interference affecting the primary user. The per-node EC in massive M2M networks was studied in \cite{eucnc} proposing three methods to alleviate interference namely power control, graceful degradation of delay constraint and the hybrid method.

Effective energy efficiency (EEE) is defined as the ratio between EC and the total consumed power \cite{EEEbook}. The maximization of EEE is of great importance for the upcoming massive M2M communication and thus the IoT, where the goal is to maximize the throughput for each consumed unit of power. In \cite{paper7}, the empty buffer probability (EBP) model was considered as an EEE booster for long packets transmission. The trade off between EEE and EC was studied in \cite{paper8} where the authors suggested an algorithm to maximize the EC subject to EEE constraint. However, the probability of transmission error that appears in finite blocklength communication due to imperfect coding was not considered. She et al. showed in \cite{paper20} that the relation between EEE and delay in wireless systems is not always a tradeoff. They concluded that a linear relation between service rate and power consumption leads to an EEE-delay non-tradeoff region.

In this paper, we derive a mathematical expression for the EEE in quasi-static Rayleigh fading for delay limited networks when applying linear power consumption model. We characterize the optimum power allocation strategy for the maximization of EEE. Afterwards, we resort to the power consumption model accounting for the probability of empty buffer in finite blocklength transmission, and prove that this model is valid for blocklength limited packets. Then, we emphasize that considering the probability of emptying the buffer during transmission of short packets allows for a more precise estimation of the EC and EEE which are higher when compared to the full buffer scenario. We depict how the optimum power allocation is affected by limiting the packet length and the performance gap that appears accordingly. Finally, we evaluate the impact of delay limitations on EC and EEE. We deduct that extending the allowable delay provides gains in EC and EEE in all cases.

The rest of the paper is organized as follows: in  Section \ref {system model}, we introduce the system model and elucidate the relation between EC and EEE. Next, Section \ref{EEE_finite blocklength} presents the EEE analysis in Rayleigh block fading scenario and characterizes optimum error and power allocation strategies for EEE maximization when applying the linear power consumption model. We illustrate the EBP model in finite blocklength transmission and characterize the EEE maximization with delay, error and power constraints in Section \ref{EBP}. The results are discussed in Section \ref{results}. Finally, Section \ref{conclusion} concludes the paper.

\section{Preliminaries} \label{system model}

Consider a point to point transmission scenario in which two nodes communicate through a Rayleigh block fading channel with blocklength $n$. The received vector $\mathbf {y}\in \mathbb{C}^n$ is given by
\begin{align}\label{eq1}
\mathbf {y}=h\mathbf {x}+\mathbf {w},
\end{align}
where $\mathbf {x} \in \mathbb{C}^n$ is the transmitted packet, and the fading coefficient is denoted by $h$ which is assumed to be quasi-static and Rayleigh distributed. This implies that $h$ remains constant over $n$ symbols which span the whole packet duration. $\mathbf{w}$ is the additive complex Gaussian noise vector whose entries are of unit variance. Furthermore, we assume that CSI is available at each node.
\subsection{Communication at Finite Blocklength}
In finite blocklength transmission, short packets are conveyed at rate that depends not only on the SNR, but also on the blocklength and the probability of error $\epsilon$ \cite{paper2}. In this case, $\epsilon$ has a small value but not vanishing. For error probability $\epsilon \in\left[ 0,1\right] $, the normalized achievable rate in bpcu is given by 
\begin{align}\label{eq3}
\begin{split}
r\approx&\log_2(1+\rho|h|^2) 
-\frac{Q^{-1}
	(\epsilon)\log_2(e)}{\sqrt{n}}\sqrt{ 1-\frac{1}{\left( 1+\rho|h|^2\right) ^{2}} } ,	
\end{split}
\end{align}
where $Q(\cdot)=\int_{\cdot}^{\infty}\frac{1}{\sqrt{2 \pi}}e^{\frac{-t^2}{2}} dt$ is the Gaussian Q-function, and $Q^{-1} (\cdot)$ represents its inverse, $\rho$ is the SNR and $|h|^2$ is the Rayleigh fading envelope coefficients, which is distributed as $f_Z(z)=e^{-z}$.  

\subsection{The relation between Effective Capacity and Effective Energy Efficiency} \label{EC_EEE} 
The concept of EC ($C_e$) indicates the capability of communication nodes to exchange data with maximum rate under statistical delay constraint. A statistical delay violation model implies that an outage occurs when a packet delay exceeds a maximum delay bound $\delta$ and the outage probability is defined as \cite{paper6}
\begin{align}\label{delay}
\Lambda=Pr(delay \geq \delta) \approx e^{-\theta. C_e. \delta},
\end{align}	
where $\Pr(\cdot)$ denotes the probability of a certain event. Conventionally, the tolerance of a system to long delays is measured by the delay exponent $\theta$. The system tolerates large delays for small values of $\theta$ (i.e., $\theta\rightarrow 0$), and it becomes stricter delay-wise for large values of $\theta$. For instance, a network with unity EC and an outage probability $\Lambda=10^{-3}$ can tolerate a maximum delay of $\delta=691$ symbol periods for $\theta=0.01$ and $\delta=23$ symbol periods when $\theta=0.3$. In quasi-static fading, the channel remains constant within each transmission period $n$ \cite{paper13}, and the EC in bits per channel use (bpcu) is \cite{paper5}
\begin{align}\label{EC}
C_e(\rho,\theta,\epsilon)=-\frac{\log\ \psi(\rho,\theta,\epsilon)}{n\theta}, 
\end{align} 
where
\begin{align}\label{psi} \psi(\rho,\theta,\epsilon)=\mathop{\mathbb{E}_{Z}}\left[\epsilon+(1-\epsilon)e^{-n\theta r}\right] ,
\end{align}
and $\log$ is the natural logarithm. In \cite{ paper5}, the effective capacity is statistically studied for single node scenario, but never to a closed form expression. It has been proven that the EC is concave in $\epsilon$ and hence, has a unique maximizer.

Indeed the raise of transmission power increases the transmission rate $r$ boosting the EC. However, this comes at the cost of energy consumption, which is not feasible for energy-limited systems such as smart grids, and
massive M2M type systems, which are the main concern of our analysis. Typically, these systems are isolated from stationary power sources. Therefore, it would be of high interest to study the energy consumption of these networks.

Defined as the ratio between EC and consumed power, the EEE metric indicates the network's capability of achieving a certain latency restricted rate with minimum energy consumption \cite{Petreska2016}. In this paper, we characterize the EEE and the optimum power allocation for different power models in the finite blocklength regime. 

\section{Linear power consumption model} \label{EEE_finite blocklength}
Consider the linear model in which power consumption is defined by \cite{Helmy}
\begin{equation}\label{Pt}
P_t(\rho)=\zeta \rho+P_c ,
\end{equation}
with $\zeta$ being the inverse drain efficiency of the transmit amplifier and $P_c$ the hardware power dissipated in circuit. For this model, the EEE is given by
\begin{equation}\label{EEE0}
\eta_{ee}=\frac{-\frac{1}{n\theta} \log\left(\mathop{\mathbb{E}_{Z}}\left[\epsilon+(1-\epsilon)e^{-n\theta r}\right]\right) }{\zeta \rho+P_c}. 
\end{equation} 
Notice that here the noise is normalized so that the SNR $\rho$ frankly represents the transmit power. This scenario assumes an always full buffer and does not account for EBP. In \cite{paper5}, a stochastic model for EC was studied, but never to a closed form expression. Here, we present a tight approximation for the EC and hence, the EEE.
\begin{lemma} \label{lemma 1}
	For a Rayleigh block fading channel with blocklength $n$, the EEE of the linear power consumption model is approximated as
	\begin{align}\label{Rayleigh}
	\begin{split}
	\eta_{ee}(\rho,\theta,\epsilon)\approx&-\frac{\log \left[ \epsilon+(1-\epsilon) \ \mathcal{J}\right]}{n \theta \left( \zeta \rho+P_c\right) },  
	\end{split}
	\end{align}	
	where 
	\begin{align}\label{c2.2}
	\begin{split}
	\mathcal{J}=e^{\frac{1}{\rho}}  \rho^\alpha  \left[\vphantom{\frac{\Gamma\left(\alpha-1,\frac{1}{\rho}\right) }{\rho^{2}}}\right. &  \left. \left( \frac{\beta^2}{2}+\beta+1\right)  \Gamma\left(\alpha+1,\frac{1}{\rho} \right) \right. \\ &\left. -\left(\frac{\beta^2}{2}+\beta\right)     \frac{\Gamma\left(\alpha-1,\frac{1}{\rho}\right) }{\rho^{2}} \right],  
	\end{split}
	\end{align} 		
	$\Gamma \left(\cdot,\cdot\right)$ is the upper incomplete gamma function \cite{Abramowitz}, $\alpha=\frac{-\theta n}{\log 2}$, $\beta=\theta \sqrt{n} Q^{-1}(\epsilon)\log_2e$, and $\gamma=\sqrt{(1-\frac{1}{(1+\rho z)^{2}})}$.
\end{lemma}
\begin{proof} Please refer to Appendix A.
\end{proof}
\begin{lemma} \label{lemma 3}
	There is a unique local and global maximizer in $\epsilon$ for the EEE in Rayleigh block fading channels which is given by
	\begin{align}\label{e*}
	\begin{split}
	\epsilon^*(\rho,c,d)\approx\arg\min_{0 \leq \epsilon \leq 1} \  \epsilon+(1-\epsilon) \ \mathcal{J}.
	\end{split}
	\end{align} 
\end{lemma}
\begin{proof}
	The expectation in (\ref{psi}) is shown to be convex in $\epsilon$ in \cite{paper5} independent of the distribution of channel coefficients $Z$. Thus, it has a unique minimizer $\epsilon^*$ which is consequently the EC maximizer given by (\ref{e*}). Presuming constant transmit power and full buffer, there is a unique maximizer in $\epsilon$ for both EC and EEE. This is because the denominator of (\ref{EEE0}) does not depend on $\epsilon$, thus maximizing EC, also maximizes the EEE. 
\end{proof}
Lemma \ref{lemma 1} provides a numerical solution for $\epsilon^*$ which can be obtained via linear search. Note that $\beta$ is not a function of $z$ or $\rho$ which simplifies the problem. The maximum effective energy efficiency $\eta_{ee_{max}}$ can be obtained by substituting $\epsilon^*$ into (\ref{Rayleigh}).

Since the logarithmic term is dominant in the rate equation given in (\ref{eq3}), it is quite straightforward to verify that the rate function has a negative second derivative for practical SNR regions and therefore is concave in power. This firmly holds for non-extremely low SNR (i.e $\geq-10$ dB) regions and the mathematical proof will be shown later in a journal version. Following a similar procedure as in \cite{paper8} based on \cite{EEE_concave}, we can conclude that the EEE in the finite blocklength regime is also a quasi-concave function of power and strictly concave in its upper contour. Hence, the optimum power allocation for maximizing the EEE is obtained by differentiating (\ref{Rayleigh}) with respect to $\rho$ as follows

\begin{align}\label{p*7}
\begin{split}
\frac{\partial \eta_{ee}}{\partial \rho}&= -\left[\frac{\frac{(1-\epsilon)\mathcal{J}^{'}(\zeta \rho+P_c)}{\epsilon+(1-\epsilon)\mathcal{J}}-\zeta \log(\epsilon+(1-\epsilon)\mathcal{J})}{n \theta(\zeta\rho+P_c)^2} \right] \\
&\approx-\left[\frac{\frac{\mathcal{J}^{'}(\zeta \rho+P_c)}{n \theta\mathcal{J}(\zeta\rho+P_c)}-\frac{\zeta \log(\epsilon+(1-\epsilon)\mathcal{J})}{n \theta(\zeta\rho+P_c)}}{(\zeta\rho+P_c)} \right]=0.
\end{split}
\end{align}
Manipulating, we obtain the optimum power allocation $\rho^*$ as the solution to
\begin{align}\label{p*9}
\begin{split}
\eta_{ee}(\rho^*)=\frac{-1}{n \theta}\left( \frac{\mathcal{J}^{'}(\rho^*)}{\mathcal{J}(\rho^*)}\right). 
\end{split}
\end{align}
Let $T_1= \frac{\beta^2}{2}+\beta+1$ and $T_2=\frac{\beta^2}{2}+\beta$ in (\ref{c2.2}), then we have
\begin{align}\label{p*10}
\begin{split}
\mathcal{J}=e^{\frac{1}{\rho}} \rho^{\alpha}\left(T_1 \Gamma\left( \alpha+1,\frac{1}{\rho}\right) -\frac{T_2}{\rho^2} \Gamma\left( \alpha-1,\frac{1}{\rho}\right).  \right) 
\end{split}
\end{align}
To differentiate $\mathcal{J}$, we apply the derivative of the upper incomplete gamma function \cite{Gradshteyn} as follows
\begin{align}\label{p*11}
\begin{split}
&\mathcal{J}^{'}=\frac{\partial \mathcal{J}}{\partial \rho} \\
&=-\frac{1}{\rho^2}\left[\mathcal{J}+ \frac{\alpha}{\rho} \mathcal{J}- \frac{T_1}{\rho^\alpha}e^{-\frac{1}{\rho}} \right. 
\left.- \frac{2 T_2}{\rho}  \Gamma\left( \alpha-1,\frac{1}{\rho}\right)+\frac{e^{-\frac{1}{\rho}}}{\rho^\alpha}  \right] \\
&=-\frac{1}{\rho^2}\left[\left(1+\frac{\alpha}{\rho} \right)\mathcal{J} +(1-T_1)\frac{e^{-\frac{1}{\rho}}}{\rho^{\alpha}}- \frac{2 T_2}{\rho}  \Gamma\left( \alpha-1,\frac{1}{\rho}\right) \right]. 
\end{split}
\end{align}
Although we could differentiate $\mathcal{J}$, a closed form solution for (\ref{p*9}) does not exist. For this purpose, we can utilize Matlab root-finding functions, e.g., fzero or plotting in a similar way to \cite{paper8}.

\section{Empty buffer probability model} \label{EBP}
\begin{figure}[!t] 
	\centering
	\includegraphics[width=0.77\columnwidth]{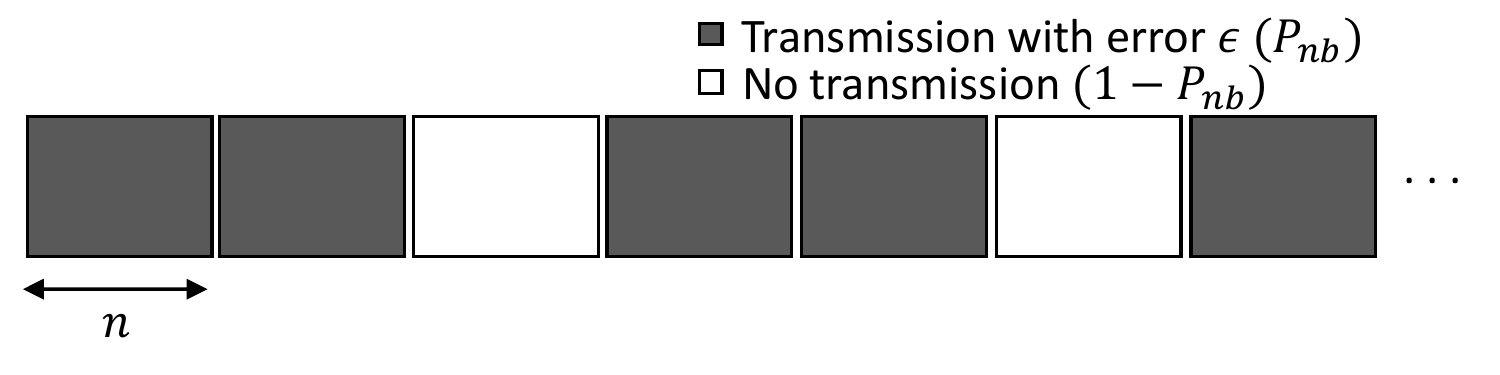}
	\vspace{-4mm}
	\caption{Transmission with empty buffer probability in quasi-static channel with blocklength $n$.}
	\label{Empty_buffer}
	\vspace{-4mm}
\end{figure}
In this section, we discuss the EEE in the finite blocklength regime for the power model which considers the EBP. Previously, we assumed that the buffer is always full which practically is not always the case. In real scenarios, there would be instants in which a certain node becomes idle and therefore has no data to transmit. Thus, we need to account for the case when the buffer is empty. Accordingly, we apply the model considered in \cite{paper7} to networks operating in the finite blocklength regime with non-vanishing probability of error $\epsilon$. We investigate the effect of EEE maximization with EBP and compare it to the case of always full buffer. After accounting for EBP, the transmission probability $P_{nb}$ is equal to ($1-$ the probability of empty buffer) and the transmission process appears in Fig. \ref{Empty_buffer}. For an arrival rate of $\lambda$ and a stable queue, the power consumption becomes 
\begin{equation}\label{Pt_nb}
P_t(\rho)=P_{nb}\zeta p+P_c=\frac{\lambda}{\mathop{\mathbb{E}}\left[r\right] }\zeta \rho+P_c,
\end{equation}
with $P_{nb}=\frac{\lambda}{\mathop{\mathbb{E}}\left[r\right] }$ denoting non-empty buffer probability (NBP). Note that $P_{nb}=1$ indicates that the buffer is always full. The EEE with EBP is given by
\begin{equation}\label{EEE1}
\eta_{ee}=\frac{-\frac{1}{n\theta}\log \left[ \epsilon+(1-\epsilon) \ \mathcal{J}\right]}{ \frac{\lambda}{\mathop{\mathbb{E}}\left[ r\right] }\zeta \rho+P_c},  
\end{equation} 
where the numerator represents the effective capacity in the finite blocklength regime.

\subsection{Verifying the effective energy efficiency model with empty buffer probability in finite blocklength}

The power consumption model considering the probability of empty buffer fulfills the characteristic properties of an energy efficiency function for the Shannon model \cite{paper7}. According to \cite{EEEbook}, an energy efficiency function must be non-negative, must be zero when the transmit power is zero, and must tend to zero as the transmit power tends to infinity. We start by verifying that this power consumption model is valid as well for short packets transmission. 
\begin{lemma} \label{EEEV}
	The EEE in (\ref{EEE1}) is zero for $\rho=0$ and tends to 0 when $\rho\rightarrow \infty$.
\end{lemma}
\begin{proof} Please refer to Appendix B.
\end{proof}

\subsection{Effective energy efficiency maximization with buffer constraints}
We investigate the EEE maximization with EC, delay, and power constraints. EC should be higher than the arrival rate $\lambda$ to guarantee a stable queue, while the transmission SNR $\rho$ is bounded by $\rho_{max}$. Thus, the optimization problem is formulated as  
\begin{equation}\label{op2}
\begin{split}
\max_{\rho \geq 0, \theta \geq 0} \ &\eta_{ee}=\frac{-\frac{1}{n\theta}\log \left[ \epsilon+(1-\epsilon) \ \mathcal{J}\right]}{P_{nb}\zeta \rho+P_c},  \\ 
s.t \ \ &C_e(\rho,\theta,\epsilon)\geq\lambda \\
&P_{nb} e^{-\theta \lambda \delta}\leq \Lambda \\
&\rho \leq \rho_{max} \\
&\epsilon \leq \epsilon_t
\end{split}
\end{equation} 
For the full buffer model, we set $P_{nb}$ to 1. We perform a line search for $\rho$ in the interval $\left[0,\rho_{max} \right]$. The optimum error probability is $\min\left[\epsilon^*,\epsilon_t \right]$ where $\epsilon^*$ is obtained from Lemma \ref{lemma 1}. When analyzing the empty buffer scenario, we set $P_{nb}=\frac{\lambda}{\mathop{\mathbb{E}}\left[ r\right] }$. Here, $\Lambda$ is the maximum allowed delay outage probability. In all cases, the optimal value of $\theta$ can be obtained from the second constraint at equality as
\begin{equation}\label{th_pb}
\theta^{*}(\rho)=\frac{1}{\lambda \delta} \log \frac{P_{nb}}{\Lambda }.
\end{equation}

\section{Results and discussion} \label{results}
In Fig. \ref{EEE}, we plot the EEE in Rayleigh block fading channel for different delay exponents using the expectation in (\ref{EEE0}) and Lemma \ref{lemma 1}. The network parameters are $n=500$ symbol periods and $\epsilon=10^{-3}$. The figure proves the accuracy of Lemma \ref{lemma 1}. Note that the EEE declines when the delay constraint becomes more strict. Furthermore, the figure shows the concavity of the upper contour of the EEE in the transmit power and the approximation in Lemma \ref{lemma 1} captures this concavity precisely. 
\begin{figure}[!t] 
	\centering
	\includegraphics[width=1\columnwidth]{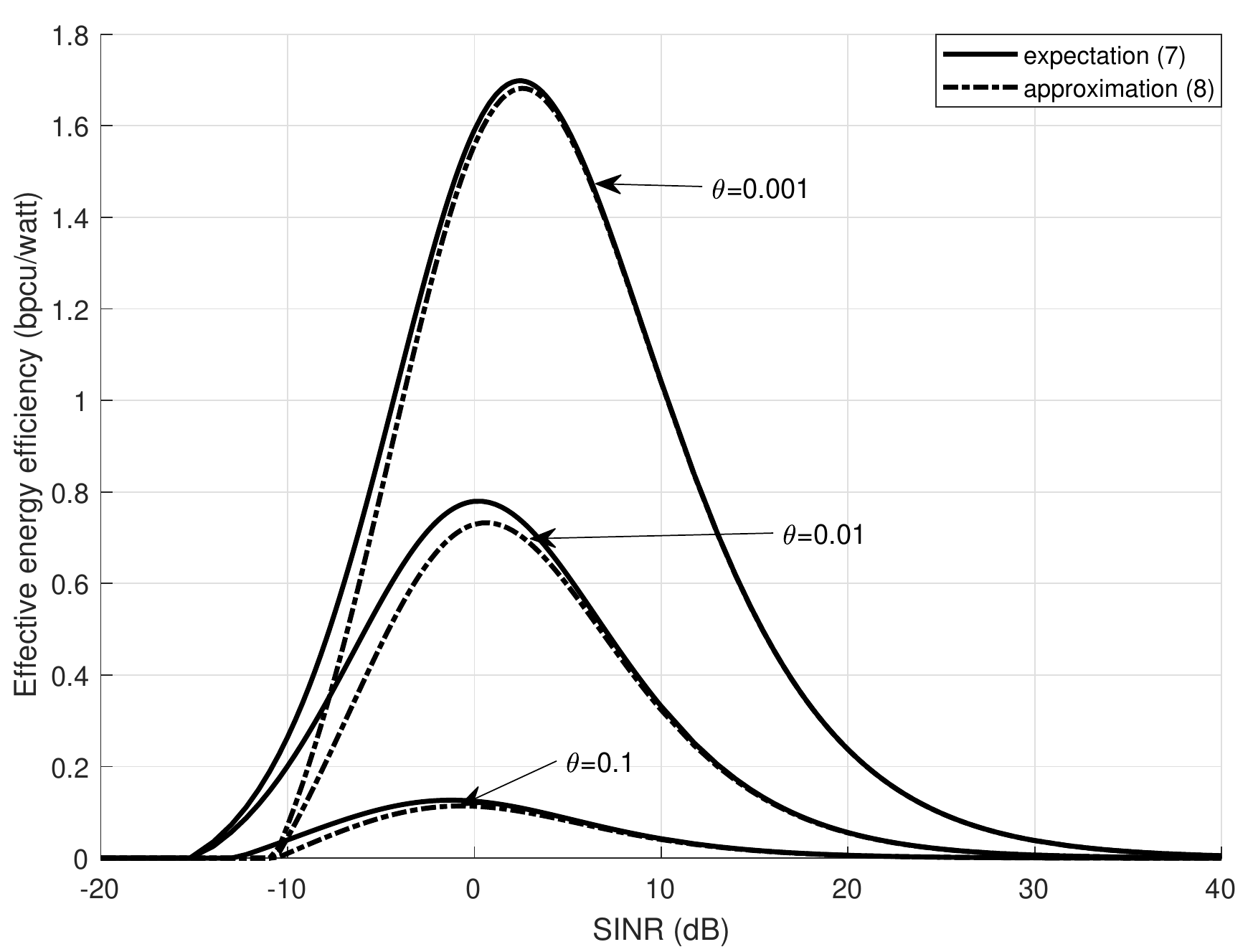}
	\centering
	\vspace{-4mm}
	\caption{Effective energy efficiency vs SNR in Rayleigh block fading for $n=500, \epsilon=10^{-3}$, $P_c=0.2, \zeta=0.2, \lambda=1$ and different delay exponents $\theta$.}
	\vspace{-4mm}
	\label{EEE}
\end{figure}

For the following simulations, we fix the network parameters as follows $\Lambda=10^{-2}, 10^{-3}, P_c=0.2 \ W, \ \zeta=0.2, \ \lambda=1, \ \delta=500$ symbol periods, and $n=500$ symbol periods, unless stated other wise. In Fig. \ref{EEE_e}, we evaluate the EEE as a function of error $\epsilon$ in case of EBP and compare it to the case where the buffer is always full while fixing the transmit power at $\rho=10$ dB. We observe that the EEE is concave in $\epsilon$ as stated in Lemma \ref{lemma 3}. It is obvious that considering the probability of empty buffer reflects a gain in the EEE over the full buffer model. Moreover, the figure depicts the EEE gap between the finite blocklength model and Shannon's bound when considering EBP where the Shannon's model considered in \cite{paper7} overestimates the EEE by more than $20 \%$ when compared to the exact finite blocklength model.  
\begin{figure}[!t] 
	\centering
	\includegraphics[width=1\columnwidth]{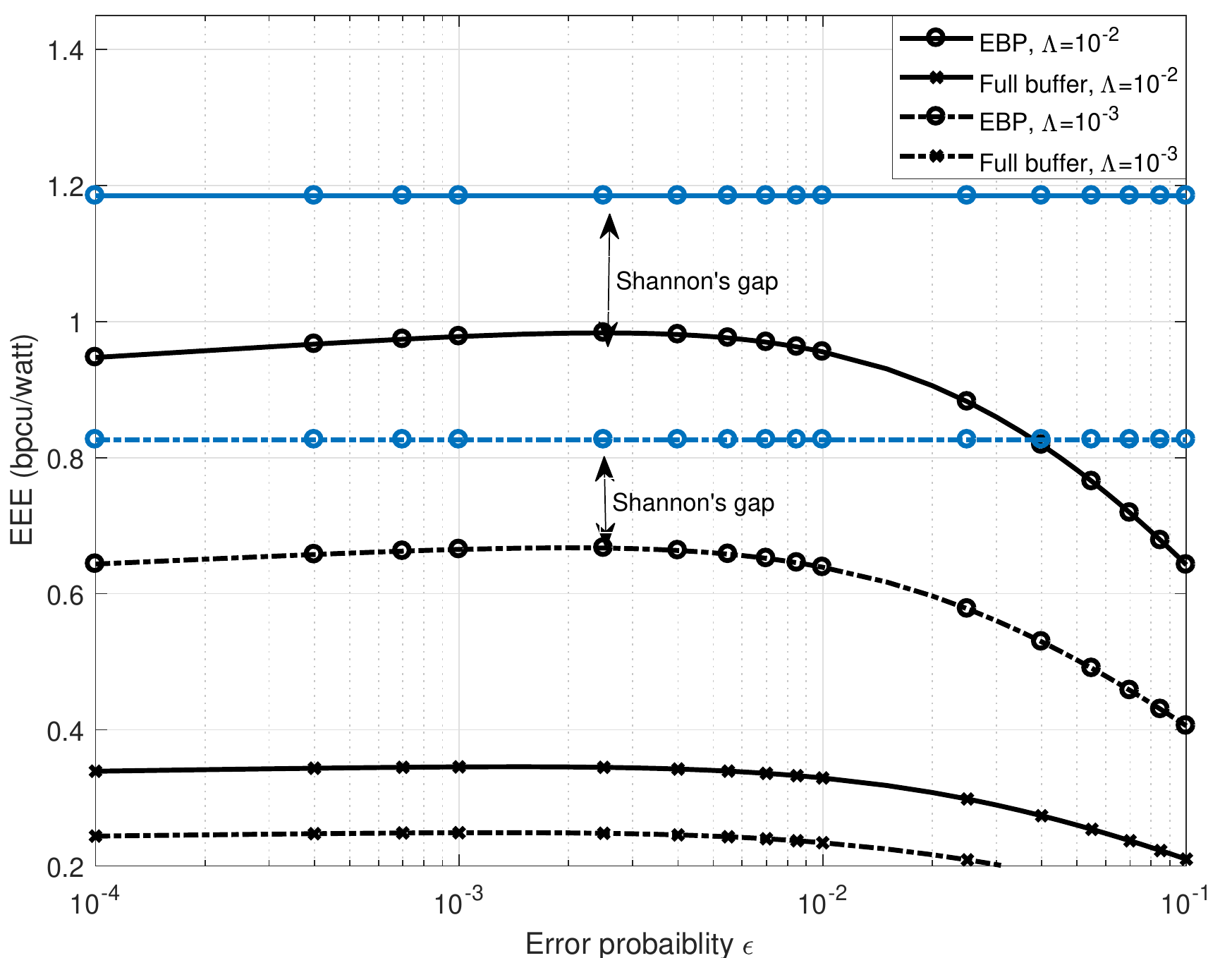}
	\centering
	\vspace{-4mm}
	\caption{EEE vs $\epsilon$ with and without empty-buffer probability for $\Lambda=10^{-2}, 10^{-3}, P_c=0.2, \zeta=0.2, \lambda=1, \delta=500$, and $n=500$.}
	\vspace{-4mm}
	\label{EEE_e}
\end{figure}

Fig. \ref{EEE_Dmax} depicts the achieved maximum EEE obtained from (\ref{op2}) for different delay limits $\delta$ where $\rho_{max}=10$ dB (variable transmission power), and $\epsilon_t=10^{-3}$. We observe that the EEE increases when extending the delay $\delta$ and relaxing the delay outage probability $\Lambda$. Again, the EEE is significantly higher when considering EBP. It is clear that the sporadic transmission scenario allows for a better modeling of the power consumption, thus is a more realistic model. This reflects that full buffer is the worst case, where we assume that all power will be consumed, while NBP models the fraction of time that is actually used for transmission of packets according to the queue congestion which highlights the gain of this model compared to always full buffer. Furthermore, the figure verifies the inaccuracy of Shannon's model when computing the EEE for relatively small packets where the inaccuracy gap reaches 0.3 bpcu/W in higher delay region.

\begin{figure}[!t] 
	\centering
	\includegraphics[width=1\columnwidth]{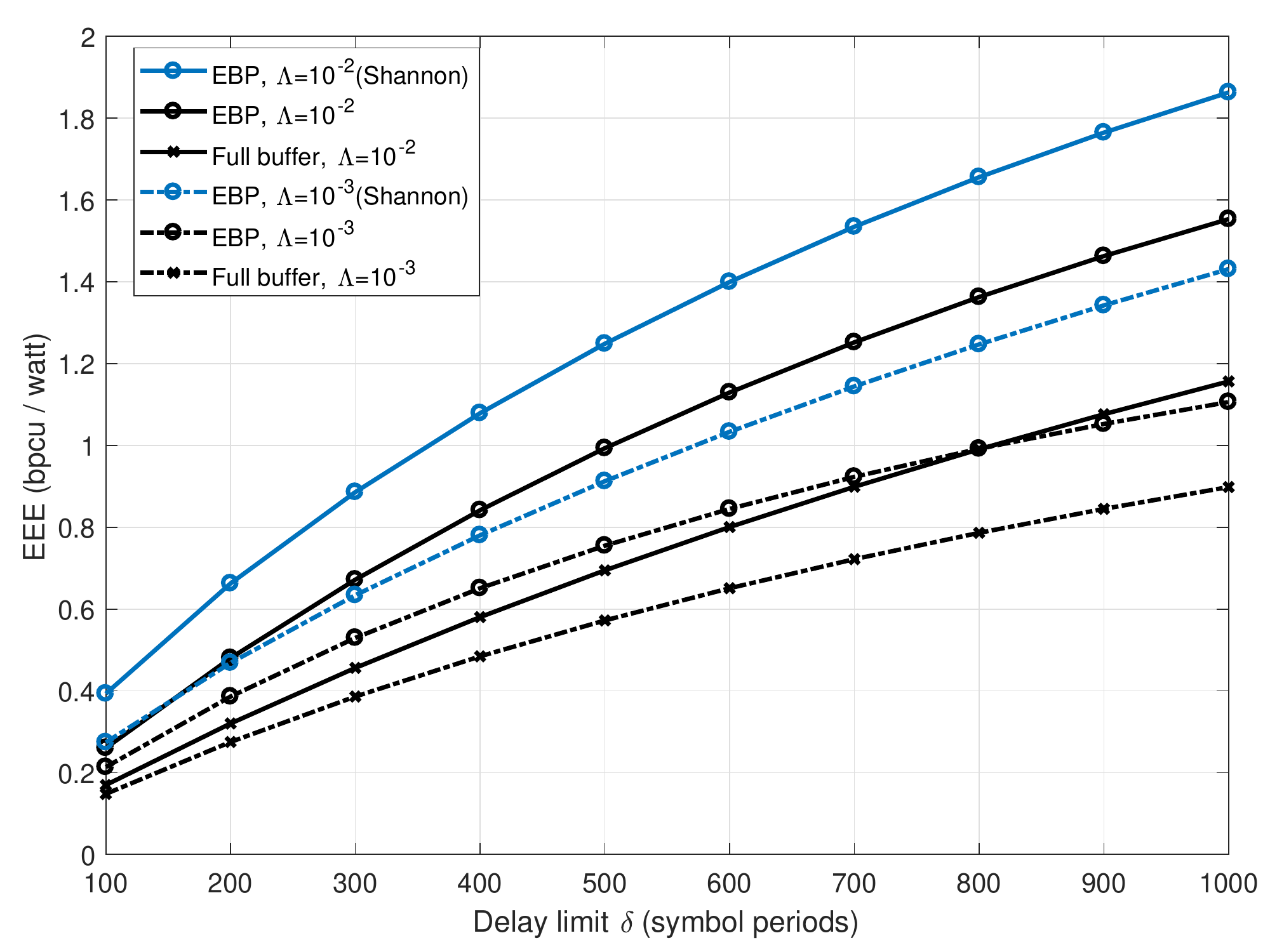}
	\centering
	\vspace{-4mm}
	\caption{EEE vs $\delta$ with and without empty buffer probability for $\Lambda=10^{-2}, 10^{-3},\ P_c=0.2 \ W, \ \zeta=0.2, \ \lambda=1$, $n=500$, and $\epsilon_t=10^{-3}$.}
	\vspace{-4mm}
	\label{EEE_Dmax}
\end{figure}

In Fig. \ref{optimal_power}, we plot the optimum power allocation for maximizing the EEE as a function of the maximum delay $\delta$ in case of EBP and always full buffer where $\rho_{max}=10$ dB. The error outage probability is fixed at $\epsilon=10^{-3}$, and the rest of the parameters are the same as in Fig. \ref{EEE_e}. The figure shows that the optimal power allocation which maximizes the EEE is significantly higher when EBP is considered. The figure also depicts that Shannon's model does not render an accurate power allocation to maximize the EEE as it underestimates the optimum power allocation when compared to the exact finite blocklength model. The power gap is ranges from $1$ to $2$ dB for the $\Lambda=10^{-2}$ as shown in the figure. Thus, we can exploit the extra power allocation that results from considering empty buffer and applying the finite blocklength model in order to boost the EC without losing energy efficiency. Moreover, the value of optimum power decays when the arrival rate $\lambda$ declines as shown for $\lambda=0.3$ in the full buffer model.

\begin{figure}[!t] 
	\centering
	\includegraphics[width=1\columnwidth]{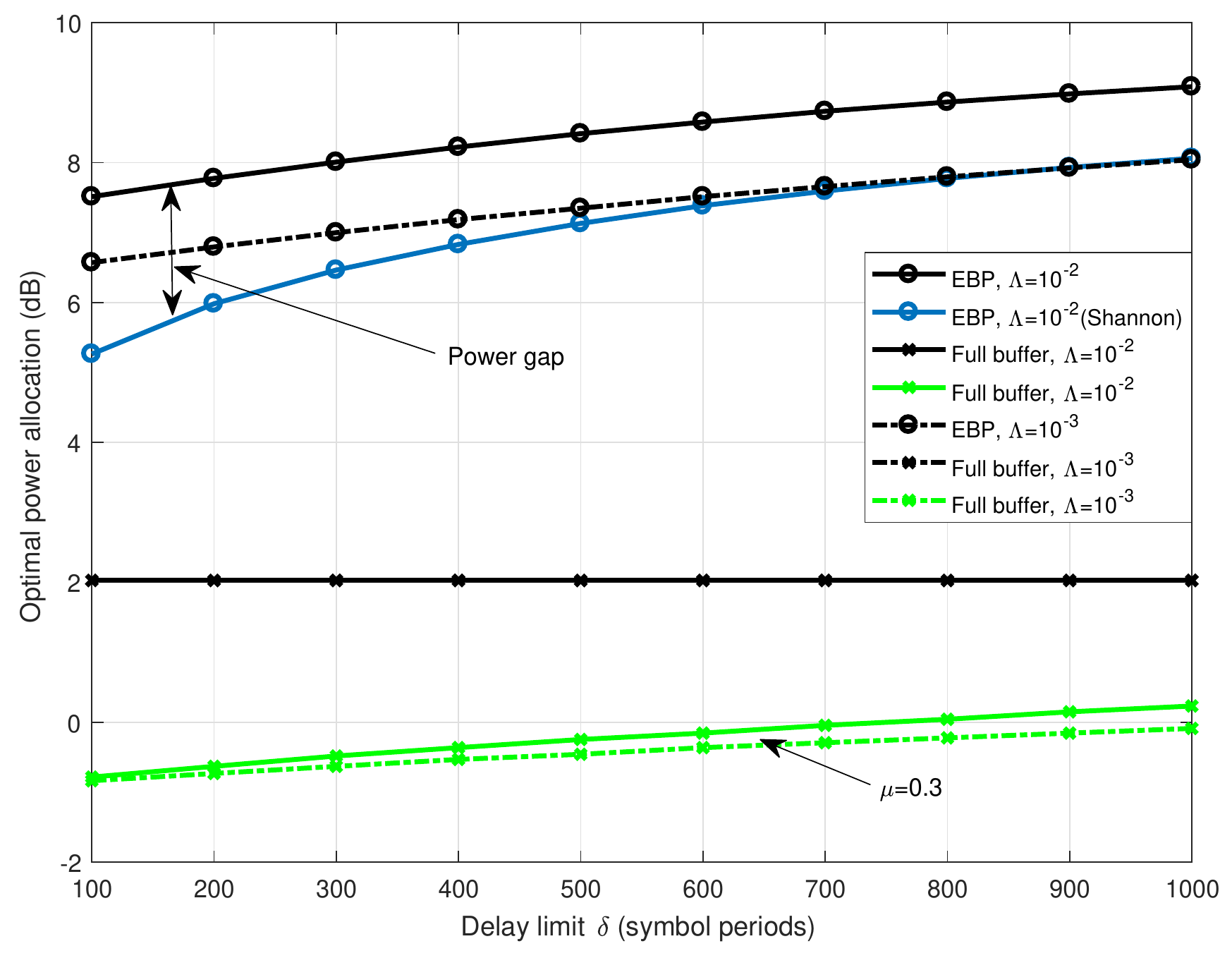}
	\centering
	\vspace{-4mm}
	\caption{Optimal power allocation vs $\delta$ with and without empty-buffer probability for $\Lambda=10^{-2}, 10^{-3}, P_c=0.2, \zeta=0.2, \lambda=1, \rho_{max}=10$ dB  and $\epsilon=10^{-3}$.}
	\label{optimal_power}
	\vspace{-4mm}
\end{figure}

Finally, we show the gain obtained in EC when considering the EBP model by plotting the EC obtained at the maximum EEE vs the maximum allowable delay $\delta$ in Fig. \ref{EC_Fig} for the same parameters as in Fig. \ref{optimal_power}. This figure shows that considering EBP not only maximizes the EEE, but also provides a significant boost in the EC when applying the optimum power allocation strategy.

\begin{figure}[!t] 
	\centering
	\includegraphics[width=1\columnwidth]{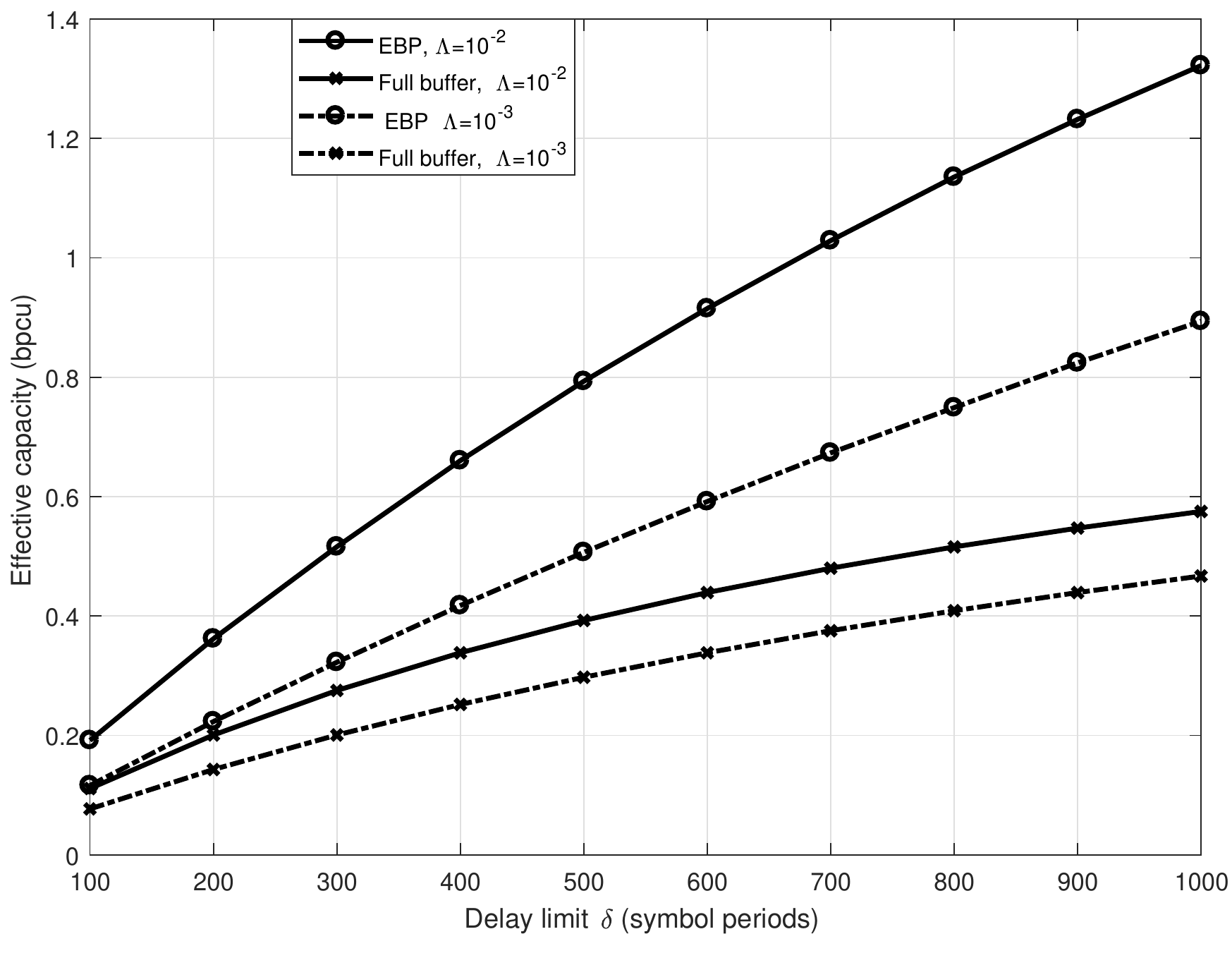}
	\centering
	\vspace{-4mm}
	\caption{Effective capacity vs $\delta$ with and without empty-buffer probability for $\Lambda=10^{-2}, 10^{-3}, \ P_c=0.2 \ W, \zeta=0.2, \ \lambda=1, \ \rho_{max}=10$ dB  and $\epsilon=10^{-3}$.}
	\vspace{-4mm}
	\label{EC_Fig}
\end{figure}

\section{Conclusion} \label{conclusion}
In this work, we presented a detailed analysis of the EEE for delay constrained networks in the finite blocklength regime. For Rayleigh block fading channels, we proposed an approximation for the EEE and characterized the EEE maximizers in terms of optimum error probability and power allocation. We showed that the advantage of considering non-empty buffer probability and flexible transmission power is twofold since it significantly improves both the EC and EEE of networks operating in the finite blocklength regime. However, Shannon's model overestimates the EEE and underestimates the optimum power allocation when compared to the exact finite blocklength model. We investigated the trade off between EEE and maximum delay limit $\delta$. The results show that in order to achieve some target latency, there is a sacrifice in EC and EEE while allowing for larger delays significantly boosts the EC and EEE.

\section*{Acknowledgments}
This work is partially supported by Aka Project SAFE (Grant no. 303532), and by Finnish Funding Agency for Technology and Innovation (Tekes), Bittium Wireless, Keysight Technologies Finland, Kyynel, MediaTek Wireless, and Nokia Solutions and Networks.

\appendices

\section{PROOF OF Lemma 1}
	Applying (\ref{psi}), we attain 
	\begin{align}\label{EC2}
	\begin{split}
	\psi(\rho,\theta,\epsilon)=\int_{0}^{\infty}
	\left( \epsilon+(1-\epsilon)e^{-\theta n r}\right)  e^{-z} \mathrm{d}z.
	\end{split}
	\end{align}
	From (\ref{eq3}), we have		
	\begin{align}\label{e1}
	e^{-\theta n r}=e^{-\theta n \log_2(1+\rho z)}e^{\theta \sqrt{n(1-\frac{1}{(1+\rho z)^{2}})} Q^{-1}(\epsilon)\log_2e}, 
	\end{align}	
	\begin{flalign}\label{e2}
	e^{-\theta n \log_2(1+\rho z)}
	&=(1+\rho z)^{\alpha}
	\end{flalign}
	which leads to
	\begin{align}\label{e3}
	e^{\theta \sqrt{n(1-\frac{1}{(1+\rho z)^{2}})} Q^{-1}(\epsilon)\log_2e}=e^{\beta \gamma}.  
	\end{align}
	We resort to the second order Taylor expansion to obtain $e^{\beta \gamma} = 1+(\beta \gamma)+\frac{(\beta \gamma)^2}{2}$. It follows from (\ref{e1}), (\ref{e2}), and (\ref{e3}) that the expression in (\ref{EC2}) can be written as
	\begin{align}\label{general2}
	\begin{split}
	&\psi(\rho,\theta,\epsilon)=  \epsilon 
	+(1-\epsilon)\left[ \int_{0}^{\infty}
	(1+\rho z)^{\alpha}e^{-z} \mathrm{d}z + \right. \\
	&\left. \beta\int_{0}^{\infty}
	(1+\rho z)^{\alpha} \gamma e^{-z} \mathrm{d}z + \frac{\beta^2}{2}\int_{0}^{\infty}
	(1+\rho z)^{\alpha} \gamma^2 e^{-z} \mathrm{d}z  \right].  
	\end{split}
	\end{align}
	The first integral can be written as $e^{\frac{1}{\rho}}  \rho^\alpha \Gamma\left(\alpha+1,\frac{1}{\rho} \right)$ By applying Laurent's expansion for $\gamma$ \cite{Complex_Analysis}, we obtain $\gamma\approx1-\frac{1}{2\left(1+\rho z \right)^2 }$. Hence, the second and third integrals can be written as  
	$e^{\frac{1}{\rho}} \beta \rho^\alpha  \left( \Gamma\left(\alpha+1,\frac{1}{\rho} \right)-\frac{\Gamma\left(\alpha-1,\frac{1}{\rho}\right) }{\rho^{2}}\right) $, and $e^{\frac{1}{\rho}} \frac{\beta^2}{2} \rho^\alpha  \left( \Gamma\left(\alpha+1,\frac{1}{\rho} \right)-\frac{\Gamma\left(\alpha-1,\frac{1}{\rho}\right) }{\rho^{2}}\right) $, respectively leading to (\ref{c2.2}).

\section{PROOF OF Lemma 3}
	For $\rho=0$, the achievable rate $r=0$ and the numerator of (\ref{EEE0}) becomes 0. Applying L'Hopital's rule for the denominator, we have   
	\begin{align}\label{lim1}
	\begin{split}
	&\lim\limits_{\rho\rightarrow 0} \frac{\rho}{\mathop{\mathbb{E}}\left[ r\right] }= \lim\limits_{\rho\rightarrow 0}\frac{1}{\mathop{\mathbb{E}}\left[ z\left(\frac{1}{(1+\rho z) \log 2}-\frac{Q^{-1}(\epsilon)\log_2(e)}{\sqrt{n}\left( 1+\rho z\right)^3 \gamma} \right)\right] }=0. \\
	\end{split}
	\end{align}
	Thus the denominator of (\ref{EEE0}) equals to $P_c$ yielding 0 for the EEE. 
	
	For the second condition, the numerator of (\ref{EEE0}) is upper bounded by $-\frac{\log \epsilon }{n \theta}$, while L'Hopital's rule for the denominator, we obtain
	\begin{align}\label{lim2}
	\begin{split}
	&\lim\limits_{\rho\rightarrow \infty}\frac{1}{\mathop{\mathbb{E}}\left[ z\left(\frac{1}{(1+\rho z) \log 2}-\frac{Q^{-1}(\epsilon)\log_2(e)}{\sqrt{n}\left( 1+\rho z\right)^3 \gamma} \right)\right] }=\infty. \\
	\end{split}
	\end{align}	
	Thus, the denominator of (\ref{EEE0}) tends to infinity which nulls the EEE. Hence, (\ref{EEE1}) holds as well under finite blocklength regime, which concludes the proof. 

\bibliographystyle{IEEEtran}
\bibliography{di}
\end{document}